%% file: ICASSP19.tex
\documentclass{article}
\usepackage{spconf,amsmath,graphicx}

\usepackage[nospace,noadjust]{cite}

\usepackage{amssymb,amsmath}
\usepackage{graphics}
\usepackage{overpic}
\usepackage{float}
\usepackage{graphicx}
\usepackage{amsmath}
\usepackage{cite}
\usepackage{color}
\usepackage{dsfont}
\usepackage{bm}
\usepackage{upgreek}
\usepackage{arydshln}
\usepackage{enumerate}
\usepackage{amsthm}
\usepackage{graphicx}
\usepackage{threeparttable}
\usepackage{caption}
\usepackage{multirow}
\usepackage{color}
\usepackage{xfrac}
\usepackage{array}
\usepackage{subfigure}
\usepackage{hyphenat}
\usepackage{scrextend}
\usepackage[toc,page]{appendix}

\newcommand{\be}{\begin{equation}}
\newcommand{\ee}{\end{equation}}

\newcommand{\norm}[1]{ || #1 ||}
\newcommand{\mb}[1]{\mathbf{#1}}
\newcommand{\bs}[1]{\boldsymbol{#1}}

\newcommand{\virg}[1]{\textquotedblleft#1\textquotedblright}

\newcommand{\tonde}[1]{\left( #1 \right)  }
\newcommand{\quadre}[1]{\left[  #1 \right]  }
\newcommand{\graffe}[1]{\left\lbrace   #1 \right\rbrace   }

\newcommand{\realpart}[1]{\mathrm{Re}\graffe{#1}}
\newcommand{\impart}[1]{\mathrm{Im}\graffe{#1}}
\newcommand{\cvec}[1]{ \mathrm{vec}\left(  #1 \right) }

\newtheorem{theorem}{Theorem}
\newtheorem{corollary}{Corollary}

\newtheorem{proposition}{Proposition}

\newtheorem{assumption}{Assumption}

\title{Scaling up MIMO radar for target detection\vspace{-0.4cm}}
\name{\vspace{-0.4cm}Stefano Fortunati, Luca Sanguinetti, Maria Sabrina Greco, Fulvio Gini\thanks{This work has been partially supported by the Air Force Office of Scientific Research under award number FA9550-17-1-0344.}}
\address{Dipartimento di Ingegneria dell'Informazione, University of Pisa, Pisa, Italy.\vspace{-0.4cm}}
\begin{document}
%
\maketitle
\begin{abstract}
This work focuses on target detection in a colocated MIMO radar system. Instead of exploiting the ``classical' temporal domain, we propose to explore the spatial dimension (i.e., number of antennas $M$) to derive asymptotic results for the detector. Specifically, we assume no a priori knowledge of the statistics of the autoregressive data generating process and propose to use a mispecified Wald-type detector, which is shown to have an asymptotic $\chi$-squared distribution as $M\to\infty$. Closed-form expressions for the probabilities of false alarm and detection are derived. Numerical results are used to validate the asymptotic analysis in the finite system regime. It turns out that, for the considered scenario, the asymptotic performance is closely matched already for $M\ge 50$.    
\end{abstract}
%
\vspace{-0.2cm}
\input{sectionI}
\vspace{-0.2cm}
\input{sectionII}
\input{sectionIII}
\input{sectionV}

\vspace{-0.3cm}
\section{Conclusions and Discussions}
\vspace{-0.2cm}
We proved that it is possible to build a Wald-type detector whose asymptotic distribution is a $\chi$-squared pdf regardless of the true, but unknown, statistical characterization of the data generating process. This was achieved by combining the misspecification theory with the paradigm of large-scale MIMO radar systems, which makes it as the first attempt to apply the \virg{massive} MIMO paradigm of communication systems to radar applications. The analysis assumed a simple autoregressive model of order $1$ for the observation data. A generalization to AR models of higher order is required to come out with a fully deployable framework. This is addressed in \cite{FSGG19} wherein a \emph{robust} Wald-type detector that does not require any a priori knowledge on the order of the autoregressive model is developed.
	\input{AppendixA}
	\bigskip
	\section{Closed-form expression for the score functions in (14) and (15) } \label{App2}
	
	Let us recall here the expressions of the \textit{score functions} introduced in Section 4 as:
	\begin{align}
	\mb{s}_1(\bs{\theta}) &\triangleq \nabla_{\bs{\theta}}\ln g(x_1|\mu_1(\bs{\theta}),s(\bs{\theta})), \quad n=1\\
	\mb{s}_n(\bs{\theta}) &\triangleq \nabla_{\bs{\theta}}  \ln g(x_n|\mu_n(\bs{\theta},x_{n-1}),\theta_5), \;n=2,\ldots,N
	\end{align}
	In the following, the closed-form expression for both is provided. From the results of Proposition 1, it is immediate to verify that: 
	\begin{align}\notag
	\ln &g(x_n|\mu_n(\bs{\theta},x_{n-1}),\theta_5) \\
		&= -\ln \theta_5 -  \theta_5^{-1} \tonde{\realpart{\varepsilon_n(\bs{\theta})}^2 + \impart{\varepsilon_n(\bs{\theta})}^2}
	\end{align}
	where
	\be
	\varepsilon_n(\bs{\theta}) \equiv \varepsilon_n(x_n,\mu_n(\bs{\theta},x_{n-1})) \triangleq x_n-\mu_n(\bs{\theta},x_{n-1})
	\ee
	for $n=2,\ldots,N$. The score vector for a single conditional observation  $\mb{s}_n(\bs{\theta}) \triangleq \nabla_{\bs{\theta}}  \ln g(x_n|\mu_n(\bs{\theta},x_{n-1}),\theta_5)$ can be expressed as:
	\be
	\label{score}
	\begin{split}
		\mb{s}_n(\bs{\theta}) & = 2 \theta_5^{-1}\tonde{\nabla_{\bs{\theta}}\realpart{\mu_n(\bs{\theta},x_{n-1})}\realpart{\varepsilon_n(\bs{\theta})}}+\\
		&+2 \theta_5^{-1}\tonde{\nabla_{\bs{\theta}}\impart{\mu_n(\bs{\theta},x_{n-1})}\impart{\varepsilon_n(\bs{\theta})}}  +\\ &+\theta_5^{-1}\quadre{\theta_5^{-1}|\varepsilon_n(\bs{\theta})|^2 - 1}\mb{e}_5
	\end{split}
	\ee
	for $n=2,\ldots,N$. Note that $\mb{e}_5 \triangleq \nabla_{\bs{\theta}}\theta_5 = [0,0,0,0,1]^T$ For $n=1$, we have:
	\be
	\label{score_1}
	\begin{split}
		\mb{s}_1(\bs{\theta}) &= 2s(\bs{\theta})^{-1} \tonde{ \nabla_{\bs{\theta}}\realpart{\mu_1(\bs{\theta})} \realpart{\varepsilon_1(\bs{\theta})}} +\\
		&+ 2s(\bs{\theta})^{-1} \tonde{ \nabla_{\bs{\theta}}\impart{\mu_1(\bs{\theta})}\impart{\varepsilon_1(\bs{\theta})}}  +\\ &+s(\bs{\theta})^{-1}\quadre{s(\bs{\theta})^{-1}|\varepsilon_1(\bs{\theta}) |^2 - 1}\nabla_{\bs{\theta}}s(\bs{\theta})
	\end{split}
	\ee
	where
	\be
	\varepsilon_1(\bs{\theta}) \equiv \varepsilon_1(x_1,\mu_1(\bs{\theta})) \triangleq x_1-\mu_1(\bs{\theta}).
	\ee
	Through straightforward calculation, the gradients involved in the previous equations can be obtained as:
	\begin{align}
	\label{grad_1}
	\nabla_{\bs{\theta}}\realpart{\mu_1(\bs{\theta})} &= (\realpart{v_1}, -\impart{v_1}, 0, 0, 0)^T,\\
	\label{grad_2}
	\nabla_{\bs{\theta}}\impart{\mu_1(\bs{\theta})} &= (\impart{v_1}, \realpart{v_1}, 0, 0, 0)^T,
	\end{align}
	\begin{figure*}
		\be
		\label{grad_3}
		\nabla_{\bs{\theta}}\realpart{\mu_n(\bs{\theta},x_{n-1})} = \left( 
		\begin{array}{c} 
			\realpart{v_n} + \impart{v_{n-1}}\theta_4-\realpart{v_{n-1}}\theta_3\\
			-\impart{v_n} + \impart{v_{n-1}}\theta_3+\realpart{v_{n-1}}\theta_4\\
			\realpart{x_{n-1}} + \impart{v_{n-1}}\theta_2-\realpart{v_{n-1}}\theta_1\\
			-\impart{x_{n-1}} + \impart{v_{n-1}}\theta_1+\realpart{v_{n-1}}\theta_2\\
			0
		\end{array}
		\right)
		\ee
		\hrule
		\be
		\label{grad_4}
		\nabla_{\bs{\theta}}\impart{\mu_n(\bs{\theta},x_{n-1})} = \left( 
		\begin{array}{c} 
			\impart{v_n} - \impart{v_{n-1}}\theta_3-\realpart{v_{n-1}}\theta_4\\
			\realpart{v_n} + \impart{v_{n-1}}\theta_4-\realpart{v_{n-1}}\theta_3\\
			\impart{x_{n-1}} - \impart{v_{n-1}}\theta_1-\realpart{v_{n-1}}\theta_2\\
			\realpart{x_{n-1}} + \impart{v_{n-1}}\theta_2-\realpart{v_{n-1}}\theta_1\\
			0
		\end{array}
		\right)
		\ee
		\hrule
	\end{figure*}
	\be
	\label{grad_5}
	\nabla_{\bs{\theta}}s(\bs{\theta}) = \left( 
	\begin{array}{c} 
		0\\
		0\\
		\frac{2\theta_3 \theta_5}{(1-\theta_3^2-\theta_4^2)^2}\\
		\frac{2\theta_4 \theta_5}{(1-\theta_3^2-\theta_4^2)^2}\\
		\frac{1}{1-\theta_3^2-\theta_4^2}
	\end{array}
	\right),
	\ee
	whereas \eqref{grad_3} and \eqref{grad_4} are shown at the top of the next page.
	\section{Closed-form expressions for the matrices $\mb{H}_1(\bar{\bs{\theta}})$ and $\mb{H}_n(\bar{\bs{\theta}})$ in (12) and (13)}\label{App3}
	
	Through direct calculation, it can be verified that:
	\begin{align}\notag
		\mb{H}_n&(\bar{\bs{\theta}})= \\\notag
		&-2[\bar{\bs{\theta}}]_5^{-1} \tonde{  \nabla_{\bs{\theta}}\realpart{\mu_n(\bar{\bs{\theta}},x_{n-1})} \nabla_{\bs{\theta}}^T\realpart{\mu_n(\bar{\bs{\theta}},x_{n-1})}} \\\notag
		& -2[\bar{\bs{\theta}}]_5^{-1} \tonde{ \nabla_{\bs{\theta}}\impart{\mu_n(\bar{\bs{\theta}},x_{n-1})} \nabla_{\bs{\theta}}^T\impart{\mu_n(\bar{\bs{\theta}},x_{n-1})}} \\
		&- [\bar{\bs{\theta}}]_5^{-2} \mb{e}_5\mb{e}_5^T
	\end{align} 
	for $n=2,\ldots,N$. Clearly, for $n=1$, we have:
	\begin{align}\notag
		\mb{H}_1(\bar{\bs{\theta}}) 
		&= -2s(\bar{\bs{\theta}})^{-1} \quadre{  \nabla_{\bs{\theta}}\realpart{\mu_1(\bar{\bs{\theta}})} \nabla_{\bs{\theta}}^T\realpart{\mu_1(\bar{\bs{\theta}})} } \\ \notag
		& + -2s(\bar{\bs{\theta}})^{-1} \quadre{\nabla_{\bs{\theta}}\impart{\mu_1(\bar{\bs{\theta}})}\nabla_{\bs{\theta}}^T\impart{\mu_1(\bar{\bs{\theta}})}}\\
		  &-s(\bar{\bs{\theta}})^{-2} \nabla_{\bs{\theta}}s(\bar{\bs{\theta}})(\nabla_{\bs{\theta}}s(\bar{\bs{\theta}}))^T.
	\end{align}
	Notice that the closed form expressions of the gradient operators involved in the previous two equations are given in \eqref{grad_1}--\eqref{grad_5}.
	
	\section{Proof of Theorem 2}\label{App4}
	
	Let us recall here the misspecified Wald (MW) detector as defined in (23):
	\begin{align}
	\label{Wald_test_1}
	\Lambda_{\mathsf{MW}}(\mb{x}) &= N{\hat{\bs{\alpha}}}^T\tonde{\mb{J}\mb{C}_N(\hat{\bs{\theta}})\mb{J}^T}^{-1}\!\!\!{\hat{\bs{\alpha}}}
	\end{align}
	with $\mb{J} = \quadre{\mb{I}_2,\mb{0}_3}$.
	
	\subsection{Asymptotic distribution of $\Lambda_{\mathsf{MW}}$ under $H_0$} 
	Next we show that, under the null hypothesis $H_0$, $\Lambda_{\mathsf{MW}}$ has an asymptotic central $\chi$-square distribution. Let us define the vector
	\be
	\bar{\bs{\theta}}_{H_0} \triangleq (\mb{0}_2,\bar{\bs{\rho}},\bar{\sigma}^2)^T
	\ee
	as the true parameter vector under the null hypothesis. Under $H_0$, from Theorem 1, we have that:
	\be
	\hat{\bs{\theta}}\overset{a.s.}{\underset{N\rightarrow \infty}{\rightarrow}}\bar{\bs{\theta}}_{H_0}
	\ee
	\be
	\label{norm_alpha}
	\sqrt{N}\tonde{\mb{J}\mb{C}_{\bar{\bs{\theta}}_{H_0}}\mb{J}^T}^{-1/2}\hat{\bs{\alpha}} \underset{N\rightarrow \infty}{\sim} \mathcal{N}\tonde{\mb{0},\mb{I}_2}.
	\ee
	For Hermitian and positive-definite matrices, the inverse operator and the principal square root are both continuous operators (see e.g. \cite{square_root_cont}), then their composition is continuous. Then, from (22) and by using the Continuous Mapping Theorem \cite[Theo. 2.7]{Billi}, we have that:
	\be
	\label{conv_C_est_root}
	\mb{C}_N^{-1/2}(\hat{\bs{\theta}}) - \mb{C}^{-1/2}_{\bar{\bs{\theta}}_{H_0}} \overset{a.s.}{\underset{N\rightarrow \infty}{\rightarrow}} \mb{0}.
	\ee
	Let us rewrite the test in \eqref{Wald_test} as in \eqref{wald_23} at the top of next page.
	\begin{figure*}
	\be
	\label{wald_23}
	\Lambda_{\mathsf{MW}}(\mb{x}) = \tonde{\sqrt{N}\tonde{\mb{J}\mb{C}_N(\hat{\bs{\theta}})\mb{J}^T}^{-1/2}\hat{\bs{\alpha}}}^T\tonde{\sqrt{N}\tonde{\mb{J}\mb{C}_N(\hat{\bs{\theta}})\mb{J}^T}^{-1/2}\hat{\bs{\alpha}}}.
	\ee
	\hrule
	\end{figure*}
	Finally, from \eqref{norm_alpha} and \eqref{conv_C_est_root}, by a direct applications of the Slutsky's Lemma, we immediately obtain \eqref{asy_H0_MWT} at the top of the next page,
	\begin{figure*}
	\be
	\label{asy_H0_MWT}
	\Lambda_{\mathsf{MW}}(\mb{x}|H_0) = {\underbrace{\tonde{\sqrt{N}\tonde{\mb{J}\mb{C}_{\bar{\bs{\theta}}_{H_0}}\mb{J}^T}^{-1/2}\hat{\bs{\alpha}}}}_{\underset{N\rightarrow \infty}{\sim} \mathcal{N}\tonde{\mb{0},\mb{I}_2}}}^T\underbrace{\tonde{\sqrt{N}\tonde{\mb{J}\mb{C}_{\bar{\bs{\theta}}_{H_0}}\mb{J}^T}^{-1/2}\hat{\bs{\alpha}}}}_{\underset{N\rightarrow \infty}{\sim} \mathcal{N}\tonde{\mb{0},\mb{I}_2}} \underset{N\rightarrow \infty}{\sim} \chi_2^2(0),
	\ee  
	\hrule
	\be
	\label{asy_H1_MWT}
	\Lambda_{\mathsf{MW}}(\mb{x}|H_1) = {\underbrace{\tonde{\sqrt{N}\tonde{\mb{J}\mb{C}_N(\hat{\bs{\theta}})\mb{J}^T}^{-1/2}\hat{\bs{\alpha}}}}_{\underset{N\rightarrow \infty}{\sim} \mathcal{N}\tonde{\tonde{\mb{J}\mb{C}_{\bar{\bs{\theta}}_{H_0}}\mb{J}^T}^{-1/2}\mb{d},\mb{I}_2}}}^T\underbrace{\tonde{\sqrt{N}\tonde{\mb{J}\mb{C}_N(\hat{\bs{\theta}})\mb{J}^T}^{-1/2}\hat{\bs{\alpha}}}}_{\underset{N\rightarrow \infty}{\sim} \mathcal{N}\tonde{\tonde{\mb{J}\mb{C}_{\bar{\bs{\theta}}_{H_0}}\mb{J}^T}^{-1/2}\mb{d},\mb{I}_2}} \underset{N\rightarrow \infty}{\sim} \chi_2^2\tonde{\mb{d}^T\tonde{\mb{J}\mb{C}_{\bar{\bs{\theta}}_{H_0}}\mb{J}^T}^{-1}\mb{d}}.
	\ee 
	\hrule
\end{figure*}
	where $\chi_2^2(0)$ indicates a central $\chi$-squared random variable with two degrees of freedom.
	
	\subsection{Asymptotic distribution of $\Lambda_{\mathsf{MW}}$ under local alternatives}
	
	Suppose that the alternatives to $H_0$ is of the form:
	\be
	\label{alter}
	H_1: \; \bs{\alpha} = \frac{\mb{d}}{\sqrt{N}}, \quad \mb{d} \in \mathbb{R}^2.
	\ee 
	Let us define the vector
	\be
	\bar{\bs{\theta}}_{H_1}^{(N)} \triangleq (\mb{d}/\sqrt{N},\bar{\bs{\rho}},\bar{\sigma}^2)^T
	\ee
	as the true parameter vector under local alternatives. Note that
	\be
	\label{lim_theta}
	\lim_{N\rightarrow\infty}\bar{\bs{\theta}}_{H_1}^{(N)} = \bar{\bs{\theta}}_{H_0}.
	\ee
	If the matrices $\mb{A}_{\bs{\theta}}$ and $\mb{B}_{\bs{\theta}}$, defined in (17) and (18) respectively, are continuous in a neighbourhood of $\bar{\bs{\theta}}_{H_0}$, then $\mb{C}_{\bs{\theta}}$ is continuous in a neighbourhood of $\bar{\bs{\theta}}_{H_0}$ as well. Then, from \eqref{lim_theta}:
	\be
	\label{lim_C}
	\lim_{N\rightarrow\infty}\mb{C}_{\bar{\bs{\theta}}_{H_1}^{(N)}} = \mb{C}_{\bar{\bs{\theta}}_{H_0}}.
	\ee
	Under $H_1$ in \eqref{alter}, from Theorem 1, we have that:
	\be
	\hat{\bs{\theta}}\overset{a.s.}{\underset{N\rightarrow \infty}{\rightarrow}}\bar{\bs{\theta}}^{(N)}_{H_1}
	\ee
	\be
	\label{norm_alpha_H1}
	\sqrt{N}\tonde{\mb{J}\mb{C}_{\bar{\bs{\theta}}^{(N)}_{H_1}}\mb{J}^T}^{-1/2}\!\!\!\!\tonde{\hat{\bs{\alpha}} - \mb{d}/\sqrt{N}} \underset{N\rightarrow \infty}{\!\!\sim\!\!} \mathcal{N}\tonde{\mb{0},\mb{I}_2}.
	\ee
	As before, by using the Continuous Mapping Theorem \cite[Theo. 2.7]{Billi} and the limiting results in \eqref{lim_C}, we have that:
	\be
	\label{conv_C_est_root_H1}
	\mb{C}_N^{-1/2}(\hat{\bs{\theta}}) - \mb{C}^{-1/2}_{\bar{\bs{\theta}}^{(N)}_{H_1}} \underset{N\rightarrow \infty}{=} \mb{C}_N^{-1/2}(\hat{\bs{\theta}}) - \mb{C}^{-1/2}_{\bar{\bs{\theta}}_{H_0}} \overset{p_{X_N}}{\underset{N\rightarrow \infty}{\rightarrow}} \mb{0}.
	\ee
	Let us rewrite the test in \eqref{Wald_test} as in \eqref{wald_23} at top of the next page.
	From \eqref{norm_alpha_H1} and \eqref{conv_C_est_root_H1}, by a direct application of the Slutsky's Lemma, we immediately obtain the result in \eqref{asy_H1_MWT} at the top of the next page.
	This result the can be used to approximate the power of the test, i.e. the $P_D$. To do this, it is enough to put $\mb{d} \equiv \sqrt{N}\bar{\bs{\alpha}}$. Then, we have:
	\be
	\Lambda_{\mathsf{MW}}(\mb{x}|H_1) \underset{N\rightarrow \infty}{\sim} \chi_2^2\tonde{N\bar{\bs{\alpha}}^T\tonde{\mb{J}\mb{C}_{\bar{\bs{\theta}}_{H_0}}\mb{J}^T}^{-1}\bar{\bs{\alpha}}}.
	\ee
	Finally, using the properties of the non-central $\chi$-squared distribution with $2$ dof \cite{non_cenrtal_CDF}, a closed form expression of the asymptotic $P_D$ can be expressed as:
	\be
	P_D(\lambda) = Q_1\tonde{\sqrt{N\bar{\bs{\alpha}}^T\tonde{\mb{J}\mb{C}_{\bar{\bs{\theta}}_{H_0}}\mb{J}^T}^{-1}\bar{\bs{\alpha}}},\sqrt{\lambda}}
	\ee
	where $Q_1(\cdot,\cdot)$ is the Marcum $Q$ function of order 1.

\bibliographystyle{IEEEtran}
\bibliography{IEEEabrv,ref_MIMO}



\end{document}

%% file: sectionI.tex

\section{Introduction}
\vspace{-0.2cm}
The first task of any multiple antenna radar system is to decide in favour of one of the two alternative hypotheses: $H_0)$ the target is absent; $H_1)$ the target is present. Given the observation vector $\mb{x}_k \in\mathbb{C}^N$ collected by the antenna array at time $k$, this detection problem can be formulated as a binary hypothesis test (HT):
\be
\label{HT_base_x}
\begin{array}[l]{clc}
	H_0: & \mb{x}_k = \mb{c}_k &  k = 1,\ldots,K \\
	H_1: & \mb{x}_k = {\alpha} \mb{v}+ \mb{c}_k & k = 1,\ldots,K
\end{array}
\ee
where $\mb{c}_k \in\mathbb{C}^N$ is the \textit{clutter} contribution. The signal of interest is $\alpha \mb{v}$, which is composed of a known vector $\mb{v}\in\mathbb{C}^N$ (called steering vector) and a deterministic, but \textit{unknown}, scalar $\alpha\in\mathbb{C}$. To discriminate between $H_0$ and $H_1$, we must define a detector $\Lambda(\mb{X})$ with $\mb{X} \triangleq [\mb{x}_1|\ldots,|\mb{x}_K]$ and then perform a test $\Lambda(\mb{X}) \overset{H_1}{\underset{H_0}{\gtrless}} \lambda$. In radar applications, $\lambda$ is usually chosen to maintain the probability of false alarm (PFA) below a pre-assigned level, say $ \overline{P_{FA}}$. Hence, $\lambda$ is computed as:
\be
\label{int_eq}
\Pr \graffe{\Lambda(\mb{X}) > \lambda | H_0} =  \overline{P_{FA}}.
\ee      
Solving the above equation is not an easy task. Moreover, the solution depends on $\Lambda(\mb{X})$ and on the statistical data model; that is, on the joint probability density function (pdf) $p_{\bs{X}}(\mb{X})$. Since $\alpha$ in \eqref{HT_base_x} is deterministic, $p_{\bs{X}}(\mb{X})$ is fully defined by the pdf $p_{\bs{C}}(\mb{C})$ of the clutter with $\mb{C} = [\mb{c}_1|\ldots|\mb{c}_K]$. In radar applications, the clutter contributions at different time instants are usually modelled as i.i.d. random vectors such that $p_{\bs{C}}(\mb{C}) = \prod\nolimits_{k=1}^{K}p_{C_N}(\mb{c}_k)$. 
In these circumstances, two popular choices for $\Lambda(\mb{X})$ are the generalized likelihood ratio (GLR) $\Lambda_{\mathsf{GLR}}(\mb{X})$ test and the Wald test $\Lambda_{\mathsf{W}}(\mb{X})$ (\cite[Ch. 9]{Cox}, \cite[Ch. 11]{kay1993fundamentalsII}). The popularity of both detectors is due to the fact that, under the hypothesis $H_0$ and $K\to \infty$, their pdfs converge to a central $\chi$-squared pdf with $2$ degrees of freedom \cite{LAN_Time_Series,asyn_time_series}. Hence, 
\eqref{int_eq} is asymptotically satisfied by $\bar{\lambda} = -2\ln \overline{P_{FA}}$. This is a particularly simple result that has received a lot of attention in the literature. However, it relies on two simplifying assumptions:
\begin{enumerate}
\item The target parameter $\alpha$ and the functional form of $p_{C_N}$ maintain constant over the observation interval.
\item The pdf $p_{C_N}$ is perfectly known.
\end{enumerate}
These two assumptions make the asymptotic analysis of $\Lambda_{\mathsf{GLR}}(\mb{X})$ and $\Lambda_{\mathsf{W}}(\mb{X})$ analytically tractable, but they are not realistic in practice. 

The main objective of this work is to develop a detector that does not rely on both assumptions while achieving the same simple asymptotic result illustrated above. Firstly, we assume $K=1$ and exploit the spatial (instead of temporal) dimension $N$ for the asymptotic analysis. This allows us to entirely drop the first assumption above. Observe that $\alpha$ remains constant over the array, while it may change over time. Secondly, we use the misspecification theory developed by Huber and White in their seminal papers \cite{Huber,white} (see also \cite{SPM}) to dispense from the knowledge of the true $p_{C_N}$. This is achieved by assuming a simpler, but \emph{misspecified}, pdf of the clutter model. Note that, unlike the classical temporal-based approach, the sample collected along the array cannot be considered as i.i.d. measurements; that is, a correlation structure has to be taken into account. Hence, the asymptotic analysis requires to use more advanced statistical tools such as those developed in \cite{white_nl_reg_dep,MMLE_dep,MMLE_time_series}. By putting together these different theories and tools we show that, under a general autoregressive assumption for the clutter generating process, it is possible to derive a \textit{misspecified} Wald-type (MW) test whose asymptotic pdf under $H_0$ is a central $\chi$-squared distribution irrespective of the true, but unknown, clutter pdf. The pdf of the proposed MW test under $H_1$ is also derived in closed form. Theoretical results will be also validated through simulations by assuming as clutter model an $\mathsf{AR(1)}$ driven by $t$-distributed innovations.

We observe that the target detection problem in large-scale radar system has been also recently discussed in \cite{Abla_conf,Abla,Large_scale}. Specifically, random matrix tools are used to get asymptotic results for the adaptive normalized matched filter for the regime in which both $M$ and $K$ go to infinity with a non-trivial ratio $M/K= c$. This is much different from this work where the temporal dimension $K$ is kept fixed.  

%% file: sectionII.tex

\section{System model}
\label{system_model}
Consider a colocated MIMO radar system equipped with $M$ transmitting elements and $M$ receiving elements \cite{Stoica_col}. Without loss of generality, we assume a transmitting and e receiving uniformly linear arrays (ULAs) of $M$ omnidirectional antennas with $d_T$ and $d_R$ element separation and a single narrowband target impinging from the angle $\phi$. This array geometry implies a transmitting and receiving steering vector of the form $\mb{a}_T(\phi) = [1, e^{j2\pi \frac{d_T}{\lambda} \sin\phi}, \ldots, e^{j2\pi\frac{d_T}{\lambda}(M-1)\sin\phi}]^T$ and $\mb{a}_R(\phi) = [1, e^{j2\pi \frac{d_R}{\lambda} \sin\phi}, \ldots, e^{j2\pi\frac{d_R}{\lambda}(M-1)\sin\phi}]^T$, respectively. The transmitted signals are obtained from $M$ orthogonal baseband waveforms through a linear transformation with $\mb{W} \in \mathbb{C}^{M \times M}$.
%
Then, the matrix $\mb{X}\in \mathbb{C}^{M \times M}$ at the output of the matched filter, for a particular range-Doppler cell, is given by:
\be
\label{sig_mod}
\mb{X} = {\alpha} \mb{a}_R(\phi)\mb{a}_T^T(\phi)\mb{W}  + \mb{C}
\ee
where $\mb{C}\in \mathbb{C}^{M \times M}$ is the disturbance matrix. In vector form, we have that
\be
\label{vec_data_model}
\mathbb{C}^{N} \ni \mb{x} = \cvec{\mb{X}} = {\alpha} \mb{v} + \mb{c}
\ee 
where $N=M^2$ and $\mb{v} = (\mb{W}^T\mb{a}_T(\phi)) \otimes \mb{a}_R(\phi)$. Notice that \eqref{vec_data_model} is in the same of \eqref{HT_base_x} when $K=1$. Note that a MIMO radar provides more degrees of freedom compared  to a phased array for which we would have $N=M$. 

As mentioned earlier, a crucial prerequisite for any radar inference task is the modelization of the \emph{clutter} contribution $\mb{c} \triangleq [c_1,\ldots,c_N]^T$. In this paper we model it according to a stationary autoregressive model of order 1, denoted as $\mathsf{AR(1)}$ as done, e.g., in \cite{Decu}. Note that, in literature, autoregressive models has been expensively used to characterize the temporal correlation (see e.g. \cite{AR_1_1,AR_1_2}, \cite[Ch. 2]{Haykin_book}).
A stationary $\mathsf{AR(1)}$ process is a discrete-time random process such that:         
\be
\label{cn}
c_n = {\rho} c_{n-1} + w_n, \quad n \in (-\infty,\infty)
\ee
where the one-lag correlation coefficient ${\rho}=\rho_R + j \rho_I=|{\rho}|e^{\mathsf{j}2\pi\nu}$ satisfies $|{\rho}|<1$ and $\{w_n: \forall  n\}$ is the so called \textit{innovation process}. We assume that $\{w_n: \forall  n\}$ are \textit{circularly symmetric} independent and identically distributed (i.i.d.) complex random variables with finite second-order moments \cite{Pici_AR,Complex_AR1} such that $w_n \sim p_w$ and $E\{|w_n|^2\}=\sigma^2 < \infty$, for all $n$. The joint pdf of $\mb{c}$ is given by 
\begin{align}
\label{AR_1_c}
p_{C_N}(\mb{c}) = p_{c_1}(c_1)\prod\nolimits_{n=2}^{N}p_w(c_n-{\rho} c_{n-1})
\end{align}
and depends on the unknown pdf $p_w$ of the innovations $\{w_n: \forall  n\}$. The autocorrelation function (ACF) of the $\mathsf{AR(1)}$ process in \eqref{cn} is given by $R[m] = \frac{\sigma^2}{1-|\rho|^2}\rho^{|m|}$
and $S(\nu) = \sigma^2\left|1-\rho e^{-\mathsf{j}2\pi\nu}\right|^{-2}$ is the relative power spectral density \cite{ARMA_AR_Kay}.
In practice, an $\mathsf{AR(1)}$ is used for modelling a directional clutter; that is, a clutter whose power is focused on a particular angular direction specified by the phase of the complex one-lag correlation coefficient $\rho$.

Most of the literature assumes that $p_w$ is Gaussian for the sake of mathematical tractability. However, this is not the case in practical radar systems where heavier-tailed models, such as the $t$-distribution, are more appropriate \cite{Sang}. As a consequence, the performance of a detection algorithm derived under Gaussian assumption are no longer reliable when the actual innovations that generates the clutter share a non-Gaussian distribution. Motivated by this consideration, we propose a Wald-type detection algorithm derived under Gaussian assumption, but with the property of having reliable and predictable performance under any $p_w$.  

%% file: sectionIII.tex

\section{HT under model misspecification}
\label{HT_missp}
The results of this paper builds upon the following assumption that characterizes a particular \textit{misspecification model}. 

\begin{assumption}[Misspecified gaussianity]\label{assumption1}
We assume that $\{c_n:\forall n\}$ in \eqref{cn} is an $\mathsf{AR(1)}$ model driven with 
\begin{align}
f_w(w_n) = (\pi\sigma^2_w)^{-1}e^{-|w_n|^2/\sigma^2_w}.
\end{align}
The true, but unknown, pdf $p_w$ is left unrestricted, except for a constraint on the finiteness of its second-order moments \cite{Pici_AR,Complex_AR1}; that is, $w_n \sim p_w$ and $E\{|w_n|^2\}={\sigma}^2 < \infty$, $\forall n$. 
\end{assumption}
For notational convenience, we denote ${\bs{\alpha}} = [\alpha_R,\alpha_I]^T$ the real representation of $\alpha$ and call ${\bs{\gamma}} = [\rho_R,\rho_I,\sigma^2]^T$
the nuisance vector. Then, we define 
\be
\bs{\theta} \triangleq \left[{\bs{\alpha}}^T,{\bs{\gamma}}^T\right]^T 
\ee
any \emph{tentative} parameter vector, while $\bar{\bs{\theta}}$ stand for the \emph{true} parameter vector underlying the data generating process. 

Under Assumption \ref{assumption1}, the misspecified pdf of the data vector $\mb{x}$ in \eqref{vec_data_model} has the parametric form $f_{X_N}(\mb{x};\bs{\theta}) = f_{C_N}(\mb{x}-\alpha \mb{v}; \bs{\gamma})$ where $f_{C_N}(\mb{c}_k;\bs{\gamma})$ is the misspecified parametric clutter pdf. 
%
%
%
The following proposition provides us with its closed-form expression (the proof can be found in Section \ref{App1}). 
\begin{proposition}
\label{Prep_fx}
If Assumption 1 holds true, $f_{X_N}(\mb{x};\bs{\theta})$ can be explicitly expressed as:
\be
\label{f_gauss}
f_{X_N}(\mb{x};\bs{\theta}) = g(x_1|\mu_1,s)\prod\nolimits_{n=2}^{N}g(x_n|\mu_n,\theta_5)
\ee
where the functional form of $g$ is $g(x|\mu,s) = \frac{1}{\pi s}e^{-\frac{|w-\mu|^2}{s}}$
with $\mu \in \mathbb{C}$ and $s \in \mathbb{R}_+$ given by 
\begin{align*} \label{mu} 
\mu_1 &= (\theta_1+j\theta_2) v_1 \\ 
\mu_n &= (\theta_1+j\theta_2) (v_n - (\theta_3+j\theta_4) v_{n-1}) + (\theta_3+j\theta_4) x_{n-1} 
\end{align*}
for $n=1, \ldots N$ and $s= \frac{\theta_5}{1-\theta^2_3-\theta^2_4}$
where $\theta_i$ indicates the $i$-th entry of the parameter vector $\bs{\theta}$ with $\dim(\bs{\theta})=5$.
\end{proposition}
Once $f_{X_N}(\mb{x};\bs{\theta})$ is computed, the next step is the implementation of a detector whose asymptotic distribution as $N\to \infty$ can be computed. A solution might be the adoption of a misspecified GLR (M-GLR) statistic, given by
\be
\label{GLR_miss}
\Lambda_{\mathsf{M-GLR}}(\mb{x}) \triangleq 2 \ln \tonde{ \frac{f_{X_N}(\mb{x};\hat{\bs{\theta}})}{\max_{\bs{\gamma}}f_{X_N}(\mb{x};[\mb{0}^T_2,\bs{\gamma}^T]^T)}}
\ee  
with $\hat{\bs{\theta}}$ being the misspecified ML (M-ML) estimate of $\bar{\bs{\theta}}$ \cite{SPM}:
\begin{align}\label{theta}
\hat{\bs{\theta}} &= \underset{\bs{\theta}}{\mathrm{argmax}} \quad f_{X_N}(\mb{x};\bs{\theta}), \quad \mb{x} \sim p_{X_N}.
\end{align}  
However, in \cite[Theo. 3.1]{Kent}, Kent proved (for the i.i.d. case) that the asymptotic distribution of $\Lambda_{\mathsf{M-GLR}}(\mb{x})$ depends on both the true $p_{X_N}$ and assumed $f_{X_N}$ data pdfs and consequently, it cannot be used to solve \eqref{int_eq} in practice.
Hence, the question is: \emph{is it possible to find a detector whose asymptotic distribution is independent of the true, but unknown, data pdf $p_{X_N}$?} The answer is positive, and the resulting detector is a misspecified Wald test. This is addressed in the following.  

%% file: sectionV.tex
\vspace{-0.2cm}
\section{Main results}
\vspace{-0.2cm}
As a prerequisite for the definition of the misspecified Wald (MW) test, we need to introduce some notation and, more importantly, to study the asymptotic properties of the MML estimator in \eqref{theta} under dependent data. Following \cite{MMLE_time_series}, we introduce 
\begin{align}
\label{H_mat_1}
\mathbb{R}^{5 \times 5} \ni\mb{H}_1(\bar{\bs{\theta}}) &\triangleq \nabla^T_{\bs{\theta}}\mb{s}_1(\bar{\bs{\theta}})\\
\label{H_mat}
\mathbb{R}^{5 \times 5} \ni\mb{H}_n(\bar{\bs{\theta}}) &\triangleq E_{p_{x_n|x_{n-1}}}\graffe{\nabla^T_{\bs{\theta}}\mb{s}_n(\bar{\bs{\theta}})|x_{n-1}}
\end{align}
for $n=2,\ldots,N$ where
\begin{align}
\mb{s}_1(\bs{\theta}) &\triangleq \nabla_{\bs{\theta}}\ln g(x_1|\mu_1(\bs{\theta}),s(\bs{\theta}))\\
\mb{s}_n(\bs{\theta}) &\triangleq \nabla_{\bs{\theta}}  \ln g(x_n|\mu_n(\bs{\theta},x_{n-1}),\theta_5)
\end{align}
are the \textit{score vectors} while $g(\cdot)$ has been defined in Proposition \ref{Prep_fx}. Under Assumption \ref{assumption1} closed form expressions of the above quantities can be computed as shown in Sections \ref{App2} and \ref{App3}. We also define 
\be
\mb{C}_{{\bar{\bs{\theta}}}} \triangleq \mb{A}^{-1}_{{\bar{\bs{\theta}}}}\mb{B}_{{\bar{\bs{\theta}}}}\mb{A}^{-1}_{{\bar{\bs{\theta}}}}
\ee
with
\begin{align}
&\mb{A}_{{\bar{\bs{\theta}}}} \triangleq \frac{1}{N}\sum\limits_{n=1}^{N}E_{p_{X_N}}\graffe{\mb{H}_n({\bar{\bs{\theta}}})}\\
&\mb{B}_{{\bar{\bs{\theta}}}} \triangleq \frac{1}{N}\sum_{n=1}^N E_{p_{X_N}}\graffe{\mb{s}_n({\bar{\bs{\theta}}})\mb{s}_n^H({\bar{\bs{\theta}}})}.
\end{align} 
Next, we provide the main results of this work for the asymptotic regime where $M\to \infty$. Clearly, this implies that $N\to\infty$ since $N=M^2$ by definition.
\subsection{Asymptotic analysis of the MML estimator }
The asymptotic properties of the MML estimator in \eqref{theta} are as follows. 
\begin{theorem}
	\label{Theo_MML_dep}
	If Assumption \ref{assumption1} holds true, the MML estimator $\hat{\bs{\theta}}$ in \eqref{theta} is consistent, i.e. $	\hat{\bs{\theta}}\overset{a.s.}{\underset{M\rightarrow \infty}{\rightarrow}}{\bar{\bs{\theta}}}$,
	and asymptotically normal, i.e.
	\begin{align}
			\label{asym_norm_gauss}
		\sqrt{N}\mb{C}^{-1/2}_{\bar{\bs{\theta}}}(\hat{\bs{\theta}}-{\bar{\bs{\theta}}}) \underset{M\rightarrow \infty}{\sim} \mathcal{N}\tonde{\mb{0},\mb{I}_5}.
	\end{align}
	Moreover, we have that
	\begin{align}
	\mb{A}_N(\hat{\bs{\theta}})\triangleq\frac{1}{N} \sum_{n=1}^{N} \mb{H}_n(	\hat{\bs{\theta}}) &\overset{p_{X_N}}{\underset{M\to \infty}{\rightarrow}} 	\mb{A}_{{\bar{\bs{\theta}}}} \\
	\mb{B}_N(\hat{\bs{\theta}})\triangleq\frac{1}{N}  \sum_{n=1}^{N} \mb{s}_n(\hat{\bs{\theta}})\mb{s}_n^H(\hat{\bs{\theta}}) &\overset{p_{X_N}}{\underset{M\to \infty}{\rightarrow}} \mb{B}_{{\bar{\bs{\theta}}}}.
	\end{align}
	such that a direct application of the Slutsky's Lemma yields
	\be
	\label{conv_C_est}
	\mb{C}_N(\hat{\bs{\theta}}) \triangleq \mb{A}_N^{-1}(\hat{\bs{\theta}})\mb{B}_N(\hat{\bs{\theta}})\mb{A}_N^{-1}(\hat{\bs{\theta}}) \overset{p_{X_N}}{\underset{M\rightarrow \infty}{\rightarrow}} \mb{C}_{{\bar{\bs{\theta}}}}.
	\ee
\end{theorem}
\begin{proof}
The proof follows from the results obtained in \cite{MMLE_dep,white_nl_reg_dep} and \cite[Theo. 2.1]{MMLE_time_series}.
\end{proof}
The main implications of Theorem \ref{Theo_MML_dep} can be summarized as follows. If the true data generating process is a stationary $\mathsf{AR}(1)$ characterized by a parameter vector $\bar{\bs{\theta}}$, and driven by i.i.d. innovations with unspecified pdf $p_w$, then the ML estimator derived under a misspecified Gaussian assumption converges (\textit{a.s.}) to the true $\bar{\bs{\theta}}$, and it is asymptotically normal independently of the true, but unknown, $p_w$. Of course, the misspecification results into a loss in terms of estimation accuracy, which is quantified by the matrix $\mb{C}_{{\bar{\bs{\theta}}}}$. 

In addition to all this, Theorem \ref{Theo_MML_dep} is the cornerstone for the derivation of the misspecified Wald test presented next.
 \vspace{-0.2cm}
\subsection{Asymptotic analysis of the misspecified Wald statistic}
Inspired by \cite{MMLE_time_series}, we consider the following MW detector
\begin{align}
\label{Wald_test}
	\Lambda_{\mathsf{MW}}(\mb{x}) &= N{\hat{\bs{\alpha}}}^T\tonde{\mb{J}\mb{C}_N(\hat{\bs{\theta}})\mb{J}^T}^{-1}\!\!\!{\hat{\bs{\alpha}}}
\end{align}
where $\mb{J} = \quadre{\mb{I}_2,\mb{0}_3}$ and $N=M^2$. The asymptotic properties of the MML estimator provided in Theorem \ref{Theo_MML_dep} are the key to study the asymptotic distributions of \eqref{Wald_test} under $H_0$ and $H_1$.
\begin{theorem}
	\label{Theo_MW}
	If Assumption \ref{assumption1} holds true, then	\begin{align}
	\Lambda_{\mathsf{MW}}(\mb{x}|H_0) \underset{M\rightarrow \infty}{\sim} \chi_2^2(0)\label{MW_H0}\\
			\Lambda_{\mathsf{MW}}(\mb{x}|H_1) \underset{M\rightarrow \infty}{\sim} \chi_2^2\tonde{\delta}
	\end{align}
	with 
	\begin{align}\label{delta}
	\delta \triangleq N\bar{\bs{\alpha}}^T\tonde{\mb{J}\mb{C}_{\bar{\bs{\theta}}_{H_0}}\mb{J}^T}^{-1}\bar{\bs{\alpha}}
	\end{align}
	and $\bar{\bs{\theta}}_{H_0} \triangleq \left[{\mb{0}}_2^T,{\bar{\bs{\gamma}}}^T\right]^T$.
\end{theorem} 
\begin{proof}
	The proof can be found in Section \ref{App4}.
\end{proof}
\vspace{-0.2cm}
Interestingly, the above theorem shows that the pdfs of \eqref{Wald_test} under $H_0$ and $H_1$ converge to $\chi$-squared pdfs with $2$ degrees of freedom when $M\to\infty$. Unlike \cite{LAN_Time_Series,asyn_time_series}, this is achieved without any a priori knowledge on the pdf of the clutter model, but relies only on the mild conditions provided in Assumption 1. In particular, it follows that \eqref{int_eq} is asymptotically satisfied by $\bar{\lambda} = -2\ln \overline{P_{FA}}$. This is valid for any pre-assigned $\overline{P_{FA}}$ and, more importantly, for any true, but unknown, pdf of the innovations. In addition to this, the following corollary can be proved.
\begin{corollary}
	\label{cor}
	If Assumption \ref{assumption1} holds true, then the probability of detection of \eqref{Wald_test} is such that
	\be
	\label{closed_form_ROC}
	P_D(\lambda) \to_{M\to\infty} Q_1\tonde{\sqrt{\delta},\sqrt{\lambda}}
	\ee 
	 where $Q_1(\cdot,\cdot)$ is the Marcum $Q$ function of order 1 and $\delta$ is given by \eqref{delta}.
\end{corollary}
Since $Q_1(\cdot,\cdot)$ is monotonic in its first argument, Corollary \ref{cor} states that the $P_D$ of the MW test in \eqref{Wald_test} goes to 1 as $M\to\infty$. Moreover, it shows that the $P_D$ depends on the true, but unknown, pdf of the innovations through the matrix $\mb{C}_{\bar{\bs{\theta}}_{H_0}}$ in \eqref{delta}. This is different from the asymptotic expression of the $P_{FA}$ that is invariant to the misspecification of $p_w$. We conclude by noticing that, even if the results of Theorem \ref{Theo_MW} and Corollary \ref{cor} are asymptotic in nature, in practice they are satisfied by \virg{practically reasonable} numbers of antennas. Indeed, in the next section numerical results are used to show that the asymptotic regime is reached already for $M = 50$.     
\vspace{-0.3cm}
\section{Numerical validation}
\vspace{-0.2cm}
Monte Carlo simulations are now used to validate the theoretical results of Theorem \ref{Theo_MW} for a MIMO radar system with a finite number of antennas. The data vector $\mb{x}$ in \eqref{vec_data_model} is generated as follows. Under the hypothesis $H_0$, we have that $\mb{x}=\mb{c}$ where $\mb{c}$ is generated according to the $\mathsf{AR(1)}$ process in \eqref{cn} with $\rho = |\rho|e^{j2\pi\nu_c}$ and $\nu_c=0.23$. The innovations $\{w_n,\forall n\}$ share a complex $t$-distribution of the form \cite{Sang,Esa}:
\be
p_{w}(w_n) = (\sigma^2 \pi )^{-1} \lambda (\lambda/\eta)^{\lambda}(\lambda/\eta + |w_n|^2/ \sigma^2)^{-(\lambda+1)}
\ee
where $\lambda \in (1,\infty)$ and $\eta = \lambda/\sigma^2(\lambda-1)$ are the shape and scale parameters. In particular, $\lambda$ controls the tails of $p_{w}$. If $\lambda$ is close to 2, then $p_{w}$ is heavy-tailed and highly non Gaussian. On the other hand, if $\lambda \rightarrow \infty$, then $p_{w}$ collapses to the Gaussian distribution. We chose $\lambda = 3$ and $\sigma^2 = 1$. Under the hypothesis $H_1$, $\mb{x}=\alpha\mb{v}+\mb{c}$ where $\mb{c}$ is generated as before, while $[\mb{v}]_n = e^{j\pi(n-1)\sin(\phi)}, n = 1,\ldots,M^2$ and $\phi = \arcsin(\nu/2)$ where $\nu = 0.25$. Note that, by choosing $\nu = 0.25$ the target comes from an angular direction, which is very close to the peak of the clutter power. 
The target term $\alpha$ is generated such that the signal-to-noise ratio is $-10\mathrm{dB}$. 


\begin{figure}[t!]\vspace{-0.8cm} 
\begin{center}
\begin{overpic}[unit=1mm,width=.95\columnwidth]{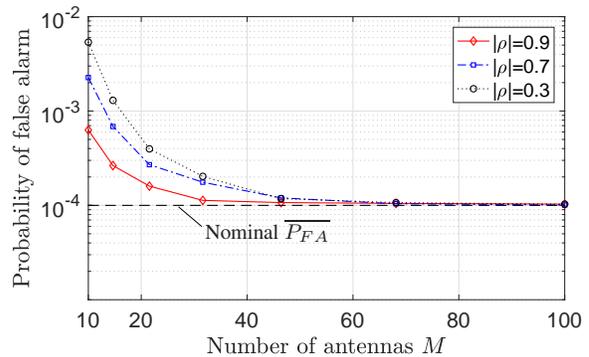}
\put(32,20){\footnotesize{Nominal $\overline{P_{FA}}$}}
\end{overpic}
\end{center}\vskip-6mm
\caption{PFA of the MW test in \eqref{Wald_test} as a function of $M$ for different values of $|\rho|$. The nominal PFA is fixed to $\overline{P_{FA}} = 10^{-4}$. Convergence to $\overline{P_{FA}}$ is achieved already for $M\ge 50$.} \vspace{-0.3cm}\label{fig:PFA} 
\end{figure}

\begin{figure}[t!]\vspace{-0.2cm} 
\begin{center}
\begin{overpic}[unit=1mm,width=.95\columnwidth]{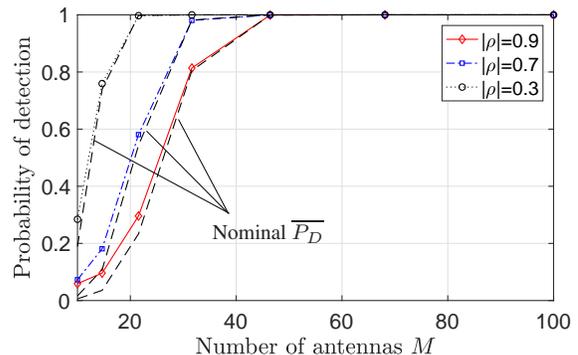}
\put(35,20){\footnotesize{Nominal $\overline{P_{D}}$}}
\end{overpic}
\end{center}\vskip-6mm
\caption{PD of the MW test in \eqref{Wald_test} as a function of $M$. Convergence to the nominal $\overline{P_{D}}$ is achieved already for $M\ge 50$.\vspace{-0.5cm} } \label{fig:PD} 
\end{figure}

Figs. \ref{fig:PFA} and \ref{fig:PD} illustrate the PFA and PD of the proposed MW test in \eqref{Wald_test} as a function of the number of antennas $M$ and different values of $|\rho|$. In line with Theorem 2, the results show that PFA tends to the nominal value $\overline{P_{FA}} = 10^{-4}$ as $M$ increases. This is achieved despite the misspecification in the clutter distribution. With $|\rho| = 0.3$, a good convergence is already achieved for $M=50$, which is a reasonable number for practical applications. A larger $M$ is needed as $|\rho|$ increases. This is an expected result since the clutter power is highly focused in the angular direction of the target. In agreement with Corollary 1, Fig. \ref{fig:PD} shows that the PD approaches the asymptotic expression provided in Corollary 1 as $M$ increases.

%% file: AppendixA.tex

\section{PROOF OF PROPOSITION 1}\label{App1}

\label{proof_prop_1}
The complex AR(1) process in $H_1$ of (1) admits a real representation as:
\be
\label{real_rep_AR}
\tilde{\mb{x}}_n = 
\left\lbrace \begin{array}{ll}
	\tilde{\bs{\mu}}_1(\alpha) + \tilde{\mb{c}}_{1}, & n=1,\\
	\tilde{\bs{\mu}}_n(\alpha,\rho,x_{n-1}) + \tilde{\mb{w}}_{n}, & n = 2,\ldots,N
\end{array}\right. 
\ee
where $\tilde{\bs{\mu}}_1(\alpha)$ and $\tilde{\bs{\mu}}_n(\alpha ,\rho, x_{n-1})$ are the real representations of the complex scalars $\mu_1(\alpha) \triangleq \alpha v_1$ and $\mu(\alpha,\rho,x_{n-1}) \triangleq \alpha (v_n -\rho v_{n-1}) + \rho x_{n-1}$, respectively. Under the (possibly misspecified) Gaussianity assumption, the pdf of the real representation of the innovations $\tilde{\mb{w}}_{n}$ is $\forall n$ 
\begin{align}
\!\!\!\!\!\!f_{\tilde{\mb{w}}_{n}}(w_{R,n},w_{I,n}) = \frac{1}{2\pi(\sigma_w^2/2)} e^{-\frac{\norm{\tilde{\mb{w}}_{n}}^2}{2(\sigma_w^2/2)}} = g(w|0,\sigma_w^2).
\end{align}
By using \eqref{real_rep_AR} and exploiting the properties of Gaussian AR(1) processes yields
\begin{align}\notag
f_{\tilde{\mb{x}}_{n}|\tilde{\mb{x}}_{n-1}}&(x_{R,n},x_{I,n}|\tilde{\mb{x}}_{n-1}) \\
&= \frac{1}{\pi \theta_5} e^{-\frac{\norm{\tilde{\mb{x}}_{n}-\tilde{\bs{\mu}}_n(\alpha,\rho,x_{n-1})}^2}{\theta_5}}\\
& = \frac{1}{\pi \theta_5} e^{-\frac{|x_{n}-\mu_n(\alpha,\rho,x_{n-1})|^2}{\theta_5}}, \; n=2,\ldots,N
\end{align}
while, from the stationarity property in (5), we get:
\be
\begin{split}
	f_{\tilde{\mb{x}}_{1}}(x_{R,1},x_{I,1}) & = \frac{1}{2\pi\tonde{\frac{\sigma^2}{2(1-|\rho|^2)}}} e^{-\frac{\norm{\tilde{\mb{x}}_{1}-\tilde{\bs{\mu}}_1(\alpha)}^2}{2\tonde{\frac{\sigma^2}{2(1-|\rho|^2)}}}}\\
	& = \frac{1}{\pi s(\alpha,\rho)} e^{-\frac{|x_{1}-\mu_1(\alpha)|^2}{s(\alpha,\rho)}}.
\end{split}
\ee
Finally, (9) follows directly from the definition of the parameter vector $\bs{\theta}$ in (8).